\documentclass[a4paper,UKenglish,cleveref, autoref, thm-restate]{lipics-v2021}

\title{Coordinating ``7 Billion Humans'' is hard} 

\author{Alessandro Panconesi}{Sapienza University of Rome, Italy}{ale@di.uniroma1.it}{https://orcid.org/0000-0002-2169-3067}{}
\author{Pietro Maria Posta}{Sapienza University of Rome, Italy}{posta.1966929@studenti.uniroma1.it}{https://orcid.org/0009-0004-2997-5538}{}
\author{Mirko Giacchini}{Sapienza University of Rome, Italy}{giacchini@di.uniroma1.it}{https://orcid.org/0009-0009-5704-098X}{}

\authorrunning{A. Panconesi, P. M. Posta, and M. Giacchini} 

\Copyright{Alessandro Panconesi, Pietro Maria Posta, and Mirko Giacchini} 

\category{FUN with Algorithms}



\nolinenumbers 

\acknowledgements{The authors wish to thank Erik Demaine, Timothy Gomez, and Della Hendrickson for the insightful discussion on the robot swarms literature.}

\EventEditors{Andrei Z. Broder and Tami Tamir}
\EventNoEds{2}
\EventLongTitle{12th International Conference on Fun with Algorithms (FUN 2024)}
\EventShortTitle{FUN 2024}
\EventAcronym{FUN}
\EventYear{2024}
\EventDate{June 4--8, 2024}
\EventLocation{Island of La Maddalena, Sardinia, Italy}
\EventLogo{}
\SeriesVolume{291}
\ArticleNo{32}

\usepackage{xspace}
\usepackage{xcolor}

\newcommand{\np}{\textrm{NP}\xspace} 
\newcommand{\pspace}{\textrm{PSPACE}\xspace} 
\newcommand{\npspace}{\textrm{NPSPACE}\xspace}

\newcommand{\cmdtext}[1]{\texttt{#1}\xspace} 
\newcommand{\stepcmd}{\cmdtext{step}}
\newcommand{\stepdircmd}{\cmdtext{step \{direction\}}}
\newcommand{\upfull}{\cmdtext{up}}
\newcommand{\downfull}{\cmdtext{down}}
\newcommand{\leftfull}{\cmdtext{left}}
\newcommand{\rightfull}{\cmdtext{right}}
\newcommand{\upc}{\texttt{u}}
\newcommand{\downc}{\texttt{d}}
\newcommand{\leftc}{\texttt{l}}
\newcommand{\rightc}{\texttt{r}}
\newcommand{\upleftc}{\texttt{(ul)}}
\newcommand{\downleftc}{\texttt{(dl)}}
\newcommand{\uprightc}{\texttt{(ur)}}
\newcommand{\downrightc}{\texttt{(dr)}}
\newcommand{\rset}{\mathcal{R}\xspace}
\newcommand{\cwset}{\mathcal{C}\xspace}
\newcommand{\num}{\texttt{num}\xspace}
\newcommand{\enc}{\texttt{enc}\xspace}
\newcommand{\stack}{\texttt{stack}\xspace}
\newcommand{\enforcesublevel}{\text{enforce\#}\xspace}

\newcommand{\overbar}[1]{\mkern 1.5mu\overline{\mkern-1.5mu#1\mkern-1.5mu}\mkern 1.5mu}

\newcommand{\problemtext}[1]{\texttt{#1}\xspace}
\newcommand{\BHessential}{\problemtext{7BH-Essential}}
\newcommand{\BHholes}{\problemtext{7BH-Holes}}
\newcommand{\intersectionnonemptiness}{Intersection Non-Emptiness Problem\xspace}

\begin{document}

\maketitle

\begin{abstract}
In the video game ``7 Billion Humans'', the player is requested to direct a group of workers to various destinations by writing a program that is executed simultaneously on each worker. While the game is quite rich and, indeed, it is considered one of the best games for beginners to learn the basics of programming, we show that even extremely simple versions are already \np-Hard or \pspace-Hard.
\end{abstract}

\section{Introduction}
In a world where robots have occupied every possible job position, humans are finally free to dedicate themselves to their favourite pastimes. However, in this utopian world, the work ethic of yesteryear reigns supreme and the only thing humans desire is good-paying jobs. To appease them, the robots construct one colossal building, so colossal to be visible from outer space, and hire all 7 billion humans living on Earth as well-paid white-collar workers. Now, robots are faced with the challenge of coordinating the human workforce to keep them all constantly entertained. 

It is in such a world that ``7 Billion Humans'' takes place. Released in 2018 as the successor of `Human Resource Machine', ``7 Billion Humans'' is a puzzle video game praised by tech reviewers as one of the best games to learn the basics of programming \cite{commonsense7BH, hackr7BH, hubspot7BH}. The player, who takes on the role of the robots in the story, must coordinate a group of workers by specifying their actions by means of an ad-hoc programming language. While the programming language in the game is quite rich, also containing "if" statements and "go-to" commands, in this paper, we will focus on the core mechanic of the game: moving the workers. In particular, the goal is to move the workers into an accepted configuration by writing a program that is executed simultaneously by each one of them. While moving, the workers will have to navigate through walls, desks, plants, and other objects that block their movement, as well as holes where workers can fall through. This extremely limited set of commands and objects constitutes the core of the game since they appear in essentially all the game levels. Even under these limitations, we show that the player (and hence the robots) will have a hard time coordinating the humans.

Our work falls within the rich area of video-game computational complexity. In recent years, several extremely popular video games have been proven to be \np-Hard or \pspace-Hard, such as Super Mario Bros. and other Nintendo games \cite{adgv15, dvw16}, Portal and several other 3D games \cite{dll16}, Trainyard \cite{alp18, alp20}, Candy Crush \cite{gln14}, and many others \cite{kps08, viglietta12}.

Let us now describe the game ``7 Billion Humans'' and our contributions. 

\subsection{Our Contributions}\label{sec:contributions}
A level of ``7 Billion Humans'' consists of a grid of cells containing \emph{workers} and \emph{objects}. The player must write a program, which is executed by \emph{every} worker simultaneously, in order to satisfy the requests of the level designer. We consider only the most basic kind of request: the workers must be moved from their starting configuration into an accepted configuration. More precisely, some cells of the grid are \emph{accepting cells} and to solve a level, all the workers must be standing on an accepting cell after executing the program. There are many commands at the disposal of the user to write the program, but we make use only of the most basic one: \stepdircmd, which is used to move the workers (all at the same time) by one cell in a given direction, which can be one of \upfull, \downfull, \leftfull, \rightfull, \upfull-\leftfull, \upfull-\rightfull, \downfull-\leftfull, or \downfull-\rightfull. A cell can either be empty or contain an object.\footnote{We can assume that workers start in empty cells and the accepting cells are all empty.} The simplest type of objects are the \emph{walls} and, as one might expect, stepping into a wall results in a non-movement.\footnote{In the game there are also other obstacles, such as desks and plants: since they all act the same, we will always talk about walls. Other workers also behave as obstacles if hit. However, this will never happen in our reductions.} We also restrict to instances where each connected component contains at most one worker.\footnote{This models those levels of the game where the workers are isolated in their own office.} More formally, for each worker $w$, let $R_w$ be the set of \emph{empty} cells that can be reached by $w$ from its starting position and imagining that there are no other workers on the grid. Then, we restrict to instances such that $R_w \cap R_{w'} = \varnothing$ for all workers $w\neq w'$.

We show in figure \ref{img:example-no-holes} an example of a level. We abbreviate the command to step in one direction with \upc, \downc, \leftc, \rightc, \upleftc, \uprightc, \downleftc, \downrightc, and a sequence of steps in the same direction using exponentiation (e.g., $\rightc^4$ instead of \rightc\rightc\rightc\rightc). For a generic finite alphabet $\Sigma$, we denote with $\Sigma^*$ the set of all finite strings consisting of symbols of $\Sigma$. A program is therefore represented as a string over the alphabet $\{\upc, \downc,\leftc,\rightc, \upleftc, \uprightc, \downleftc, \downrightc\}$. For simplicity, when a program $\pi\in\{\upc, \downc,\leftc,\rightc, \upleftc, \uprightc, \downleftc, \downrightc\}^*$ solves a level, we also say that the level \emph{accepts} the string $\pi$.

\begin{figure}[ht]
\centering
\begin{subfigure}[b]{\textwidth}
     \centering
     \includegraphics[width=\textwidth]{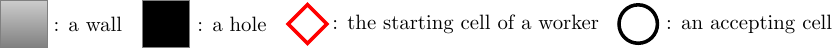}
 \end{subfigure}
 \begin{subfigure}[b]{0.49\textwidth}
     \centering
     \includegraphics[width=0.6\textwidth]{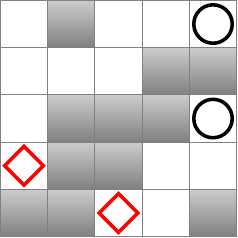}
     \caption{This level can be solved, for example, by the program $\upc^2\rightc^2\upc\rightc^2\upc$ or by the program $\rightc\upc^2\rightc^2\upc\rightc^2$. Instead, $\rightc\upc\rightc\upc$ does not solve the level because only one worker reaches an accepting cell.}
     \label{img:example-no-holes}
 \end{subfigure}
 \hfill
 \begin{subfigure}[b]{0.49\textwidth}
     \centering
     \includegraphics[width=0.6\textwidth]{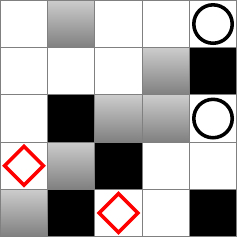}
     \caption{This level can be solved by the program $\rightc\upc^2\rightc^2\upc\rightc^2$. Instead, the program $\upc^2\rightc^2\upc\rightc^2\upc$ does not work anymore because one of the two workers would get stuck in a hole.}
     \label{img:example-holes}
 \end{subfigure}
\caption{Two examples of game levels. The level on the left contains only walls and empty cells, while the one on the right also contains holes. We assume that the grids are surrounded by walls. Note that there are two connected components, each with a single worker.}
\label{img:example-level}
\end{figure}

The decision problem that we consider is: given a level of ``7 Billion Humans'', say if the level is solvable or not. Since we included only the essential elements of the game, we call this problem \BHessential. We will show that even this extreme simplification is already \np-Hard.

\begin{theorem}
\label{thm:np-hard}
It is \np-Hard to check if a given level of ``7 Billion Humans'' is solvable, even using only \emph{walls}, \emph{empty} cells, the \stepcmd command and restricting to instances where each connected component contains at most one worker. That is, \BHessential is \np-Hard.
\end{theorem}
\emph{Holes} are another common object in the game. When a worker steps on a cell containing a \emph{hole}, it gets stuck for the rest of the computation. Since our goal is to have each worker on an accepting cell, even if a single worker steps on a hole, the level is lost. An example of a level with holes is shown in figure \ref{img:example-holes}. We call \BHholes the decision problem previously described where we add \emph{holes} in addition to the already mentioned elements. Adding holes might make the problem more difficult, in fact, \BHholes is \pspace-Hard.

\begin{theorem}
\label{thm:pspace-comp}
Determining whether a given level of ``7 Billion Humans'' is solvable is \pspace-Complete, where the level contains only \emph{walls}, \emph{holes}, \emph{empty} cells, uses only the \stepcmd command and where each connected component contains at most one worker. That is, \BHholes is \pspace-Complete.
\end{theorem}

In the literature, it is known that allowing multiple workers in the same connected component, rather than including holes, might make the problem harder as well. Note that workers block each other when hit. The following is an immediate consequence of \cite[Theorem 2]{ccgls20}:\footnote{Consider the same reduction described in figure 3b of \cite{ccgls20}. Set every cell as accepting except for the three (or, in general, $k$) bottom-left cells. This proves Theorem \ref{thm:multiple-workers}.}

\begin{theorem}[\cite{ccgls20}, Theorem 2]
\label{thm:multiple-workers}
It is \pspace-Complete to check if a given level of ``7 Billion Humans'' is solvable, where the level contains only \emph{walls}, \emph{empty} cells, uses only the \stepcmd command and allows multiple workers in the same connected component.
\end{theorem}

The paper is organized as follows.
In Section \ref{sec:related-probs} we highlight some connections between ours and other problems. In Sections \ref{sec:np} and \ref{sec:pspace} we prove our two main theorems (Theorem \ref{thm:np-hard} and \ref{thm:pspace-comp}) and finally, in Section \ref{sec:conclusion}, we discuss some final remarks.

\section{Relations with other problems}\label{sec:related-probs}
In our reductions, the game level of ``7 Billion Humans'' will be divided into isolated sub-levels, each containing a single worker. Then, to solve the level, all the sub-levels must be solved simultaneously. Moreover, in our reductions, we will prevent diagonal movements: the only admissible directions are \rightfull, \leftfull, \upfull, and \downfull. This special structure of the level can be used to draw some relations with other problems.

\subsection{Robot Swarms}
If we consider the problem without holes, the problem can be described exactly using the robot-swarm terminology \cite{ccgls20}. In particular, \BHessential is equivalent to ``$1\times 1$ $k$-region relocation with uniform control signals in the single-step model''. In such problem, there are $k$ disjoint regions of a square grid, each containing a single $1\times 1$ robot. All robots move in the same direction by a single step (if they can), and the goal is for each robot to reach one destination cell. Therefore, Theorem \ref{thm:np-hard} implies that such robot-swarm problem is \np-Hard. 

It is known in the literature that if multiple robots can be in the same region, the problem is \pspace-Complete \cite[Theorem 2]{ccgls20}. The problem is \pspace-Complete also if robots move maximally rather than by a single step (this is called full-tilt model) \cite[Theorem 8]{ccgls20}. However, the reductions used for these two problems are very different from ours, and adapting them seems non-trivial. To the best of our knowledge, it is unknown whether \BHessential (and the corresponding robot-swarm problem) is \pspace-Hard.

\subsection{Simultaneous Maze Solving}
Each sub-level can be interpreted as a grid maze: the worker must find a path from its starting position to one of the accepting cells,\footnote{Note that passing upon an accepting cell is not enough: each worker must be standing on an accepting cell at the end of the sequence of moves.} avoiding the holes and navigating through the walls. Our results, then, entail that solving multiple mazes simultaneously is \np-Hard (Theorem \ref{thm:np-hard}), or \pspace-Hard if the mazes can contain holes (Theorem \ref{thm:pspace-comp}). The only other work on the topic, to the best of our knowledge, is the one of Funke et al. \cite{fns17}, which, however, studies very special mazes that are always solvable simultaneously.

\subsection{Intersection Non-Emptiness}
Each sub-level can also be interpreted as a deterministic finite automaton (DFA for short). In particular, a sub-level $w \times h$ naturally translates into a DFA with at most $w\cdot h$ states (corresponding to the cells), and with the transition function on the alphabet $\{\upc, \downc, \leftc, \rightc\}$ that simulates the behavior of the cells. Solving all the sub-levels is equivalent to finding a string that is accepted by a set of DFAs: this is a fundamental problem in automata theory known as \intersectionnonemptiness \cite{afhhjow21, lr92, wehar14} and first shown to be \pspace-Complete by Kozen \cite{kozen77}. The structure of our DFAs is very special: the undirected transition graph, excluding self-loops, is a subgraph of the $w\times h$ grid graph. Therefore, our Theorem \ref{thm:pspace-comp} entails that Intersection Non-Emptiness is \pspace-Complete even with this class of automata.

To the best of our knowledge, given the strong restrictions that we have to make on the DFAs, our results are not derivable from existing work. As an example, in the original \pspace-Hardness proof of Kozen \cite{kozen77}, the standard construction of the DFAs contains vertices having in-degree $(|Q|+|\Sigma|+1)^3$, where $Q$ and $\Sigma$ are the set of states and input symbols of a Turing Machine, therefore, the in-degree is at least $27$ and possibly much larger. Instead, our DFAs have an in-degree of at most 8, considering self-loops. Moreover, in such proof, the automata use several one-way transitions; instead, the transitions of our DFAs are reversible (except for states associated with holes).

\section{\np-Hardness of \BHessential}\label{sec:np}
In this section we prove Theorem \ref{thm:np-hard}. In particular, we show a polynomial-time reduction from Positive 1-in-3-SAT, notoriously known to be \np-Hard (see, e.g., \cite[page 259, problem LO4]{gj79}), to \BHessential. 
\begin{definition}[Positive 1-in-3-SAT]
The input of Positive 1-in-3-SAT consists of $n$ boolean variables, $x_1, x_2,\dots, x_n$, and a set of $m$ clauses, each containing exactly three distinct positive variables (i.e., there are no negated literals). The goal is to find a truth assignment to the variables such that each clause contains exactly one true variable.
\end{definition}

Fixed an instance of Positive 1-in-3-SAT, we make use of three types of gadgets: (i) the \emph{diagonal} gadget, to prevent diagonal movements, (ii) the \emph{assignment gadget} that, intuitively, allows assigning truth values to the variables, and (iii) one \emph{clause gadget} for each clause, to ensure that the truth assignment satisfies the original 1-in-3 formula. Each gadget will be built with multiple independent sub-levels that must be solved simultaneously. Each sub-level will contain exactly one worker. The final game level is obtained by stacking the sub-levels together and isolating them via walls.

\begin{figure}[h!]
\centering
\includegraphics[height=0.18\textwidth]{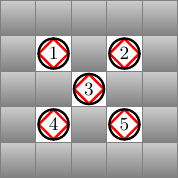}
\caption{The \emph{diagonal} gadget, used to prevent diagonal movements, consists of five sub-levels: for $i\in\{1,2,3,4, 5\}$, the $i$-th sub-level contains only the worker and the accepting cell labeled with $i$.}
\label{img:np-diagonal-gadget}
\end{figure}

The diagonal gadget is reported in figure \ref{img:np-diagonal-gadget}. It consists of five sub-levels that we draw together for brevity. In particular, for $i\in\{1,2,3,4, 5\}$, the $i$-th sub-level contains only the one worker and the one accepting cell labeled with $i$. Suppose the program contains a diagonal movement (i.e., \upleftc, \uprightc, \downleftc, or \downrightc), then, worker 3 would "overlap"\footnote{the workers are in different sub-levels, so they "overlap" if we imagine the sub-levels on top of each other} with another worker and it would become impossible to solve all the five sub-levels of the gadget (indeed, once two workers are overlapped in different sub-levels having the same walls, it is impossible to separate them). 

\subsection{Assignment Gadget}
This gadget consists of four sub-levels, reported together in figure \ref{img:np-assignment-gadget}. For $i\in\{1,2,3,4\}$, the $i$-th sub-level contains only the $i$-th worker and the accepting cells labeled with $i$. 

\begin{figure}[ht]
\centering
\includegraphics[width=0.5\textwidth]{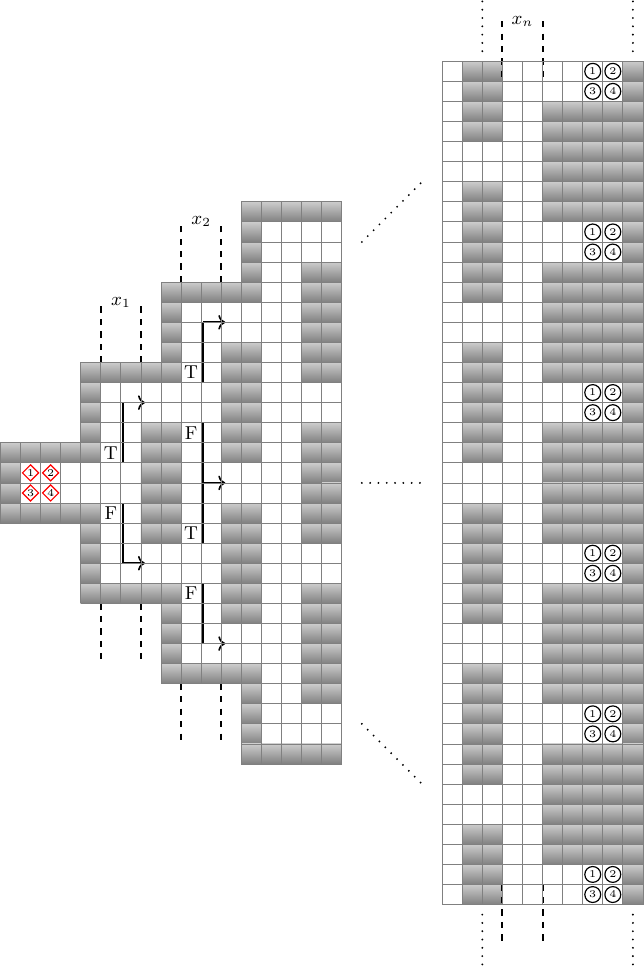}
\caption{Assignment Gadget. It consists of four sub-levels, drawn together for brevity. In particular, the $i$-th sub-level, for $i\in\{1,2,3,4\}$, contains only the worker with label $i$ on the left, and only the accepting cells with label $i$ on the right. The workers select the truth value of the variables by moving \upfull or \downfull.}
\label{img:np-assignment-gadget}
\end{figure}

By using four sub-levels, the gadget ensures that the workers cannot hit the walls, indeed, if that happened, two workers would overlap and at least one sub-level would become unsolvable (e.g., using the program $\rightc^5$, workers 1 and 2 overlap, making it impossible for them to simultaneously stand on an accepting cell). This property will be used in the clause gadgets.

Since the accepting cells are at the right end of the sub-levels, the workers must pass through all the variables, and select their truth values by moving \upfull or \downfull. Specifically, when the workers are in correspondence with the variable $x_i$ (that is, when the workers are in columns $4i+1$ and $4i+2$, assuming that the columns are numbered from left to right starting with 0 on the left wall), moving \upfull will set $x_i$ to true in all clauses while moving \downfull will set it to false.
Intuitively, we can think that the workers will move via a program of the form $\rightc^4\sigma_1^4\rightc^4\dots\sigma_n^4\rightc^4$ with $\sigma_i\in\{\upc,\downc\}$, and $\sigma_i$ determines the truth value of the $i$-th variable ($\upc=$ True, $\downc=$ False).

\begin{remark*}
Note that the workers have more freedom of movement than what we would like. For example, they can move \leftfull, and in correspondence with a variable, they can move \upfull and \downfull multiple times before going \rightfull. However, our clause gadgets are such that this extra freedom does not permit cheating. More precisely, if the level can be solved, then it can also be solved by a program of the form $\rightc^4\sigma_1^4\rightc^4\dots\sigma_n^4\rightc^4$ with $\sigma_i\in\{\upc,\downc\}$.
\end{remark*}

\subsection{Clause Gadgets}
For each clause $C=(x_a, x_b, x_c)$, with $a<b<c$, we create a clause gadget consisting of two sub-levels. In figure \ref{img:np-general-clause-gadget}, we report the sub-levels for the generic clause $C$. For the reader's convenience, we also show, in figure \ref{img:np-example-clause-gadget}, the gadget for the clause $(x_1, x_2, x_3)$. For $i\in\{1,2\}$, the $i$-th sub-level contains only the $i$-th worker and the accepting cells labeled with $i$. 

\begin{figure}[ht]
\centering
\begin{subfigure}[b]{\textwidth} 
     \centering
     \includegraphics[width=0.82\textwidth]{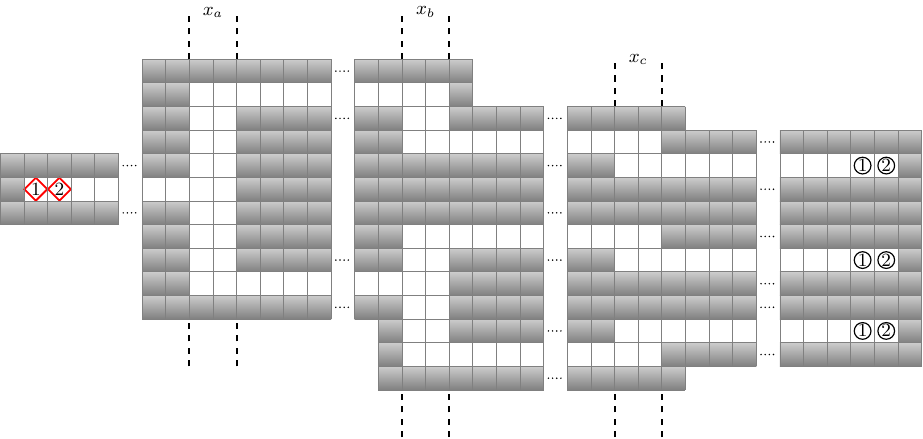}
     \caption{Gadget for the generic clause $(x_a, x_b, x_c)$, with $1\leq a < b < c \leq n$.} 
     \label{img:np-general-clause-gadget}
\end{subfigure}
\begin{subfigure}[b]{\textwidth} 
     \centering
     \includegraphics[width=0.62\textwidth]{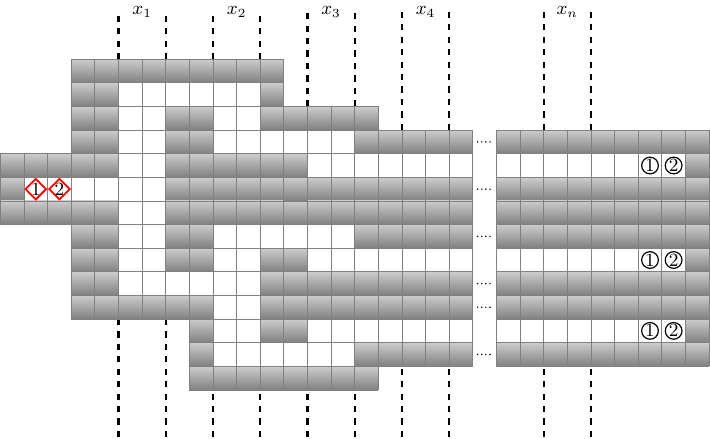}
     \caption{Gadget for the clause $(x_1, x_2, x_3)$. We can handle all the clauses by stretching opportune sections of the gadget (see figure \ref{img:np-general-clause-gadget}).}
     \label{img:np-example-clause-gadget}
\end{subfigure}

\caption{Figure (a) shows the Clause Gadget for the generic clause $(x_a, x_b, x_c)$. It consists of two sub-levels drawn together for brevity. The $i$-th sub-level, $i\in\{1,2\}$, contains only the worker with label $i$ on the left and the accepting cells with label $i$ on the right. The columns associated with the variables are aligned with those of the assignment gadget (figure \ref{img:np-assignment-gadget}). Solely for the sake of clarity, we also show in figure (b) the gadget for the particular clause $(x_1, x_2, x_3)$.}
\label{img:np-clause-gadget}
\end{figure}

Let us first argue that the flexibility of the \emph{assignment gadget} cannot be exploited to generate invalid truth assignments. Observe that the two workers of the clause gadget (of the two distinct sub-levels) must always remain side by side, indeed, if they were to split, the assignment gadget would become unsolvable since a worker would have hit a wall. Given that the two workers must remain side by side, the workers in the clause gadget cannot hit a wall horizontally (otherwise they would overlap, making the level unsolvable). This means that using the \leftfull command is useless: in fact, it can only be used to go back to a variable and possibly change its truth value in \emph{all} the clauses. Moreover, moving \upfull and \downfull multiple times when selecting a truth value does not lead to benefits either. Indeed, assuming that the user moves \leftfull or \rightfull only when none of the workers would hit a wall, then only the last \upfull or \downfull movement is used to decide the truth value of the variable in \emph{all} the clauses. In other words, if the final game level is solvable, then it is also solvable by a program of the form $\rightc^4\sigma_1^4\rightc^4\dots\sigma_n^4\rightc^4$ with $\sigma_i\in\{\upc,\downc\}$.

Now, observe that the workers of the clause gadget for $C=(x_a, x_b, x_c)$ arrive at an accepting cell if and only if they move \upfull exactly once in correspondence with $x_a$, $x_b$, or $x_c$. Note that the vertical movement in correspondence with the other variables is ignored. Since the \upfull movement is associated with the ``True'' truth value, the clause gadget behaves exactly like a clause of the Positive 1-in-3-SAT formula.

\subsection{Conclusion of the proof}
\begin{proof}[Proof of Theorem \ref{thm:np-hard}]
Suppose the Positive 1-in-3-SAT instance is solvable, and let $\alpha_i\in\{\text{True}, \text{False}\}$ for $i\in\{1,\dots,n\}$ be the truth assignment to the variables that solves the instance. Then, letting $\sigma_i=\begin{cases}\upc, \text{ if }\alpha_i=\text{True} \\\downc, \text{ if }\alpha_i=\text{False}\end{cases}$ the program $\rightc^4\sigma_1^4\rightc^4\dots\sigma_n^4\rightc^4$ solves the corresponding \BHessential instance. Conversely, if the \BHessential instance is solvable, we can assume without loss of generality that it is solved by a program of the form $\rightc^4\sigma_1^4\rightc^4\dots\sigma_n^4\rightc^4$ with $\sigma_i\in\{\upc,\downc\}$. Let $\alpha_i = \begin{cases}\text{True}, \text{ if }\sigma_i=\upc \\ \text{False}, \text{ if }\sigma_i=\downc\end{cases}$ then, as we argued, $\alpha$ satisfies all the 1-in-3-SAT clauses. 

Moreover, note that each sub-level of the assignment gadget contains $O(n^2)$ cells, each sub-level of a clause gadget contains $O(n)$ cells, and the five sub-levels of the diagonal gadget have a constant number of cells. Therefore, the complete game level obtained by stacking all the sub-levels contains $O(n^2 + nm)$ cells and can be constructed in polynomial time. Our reduction is complete.
\end{proof}

\section{\pspace-Completeness of \BHholes}\label{sec:pspace}
In this section, we prove Theorem \ref{thm:pspace-comp}. For a positive integer $z$, we use the notation $[z]=\{1,2,\dots, z\}$. Let us start by showing that \BHholes can be solved in polynomial space. 

\begin{observation}\label{obs:membership-pspace}
$\BHholes\in \pspace$
\end{observation}
\begin{proof}
Consider a game level $n\times m$ containing $k$ workers. Since the workers are the only non-static elements of the level, each cell can be in at most one of two states: either it contains a worker, or it does not. Then, there are at most $\binom{nm}{k} \leq 2^{nm}$ possible configurations for the level. Then, consider the non-deterministic Turing Machine $M$ that, given in input a level of \BHholes, maintains the current configuration of the level and a counter of the number of steps performed so far. If all the workers are standing on an accepting cell, then $M$ accepts; if instead the counter exceeds $2^{nm}$, then $M$ rejects. In all other cases, $M$ guesses non-deterministically the next step in $\{\leftc,\rightc,\upc,\downc, \upleftc, \uprightc, \downleftc, \downrightc\}$, updating the configuration and increasing the counter accordingly. It is evident that $M$ solves \BHholes because if the level is solvable, there is a solution using at most $2^{nm}$ instructions. Moreover, the space required in any computation branch is $O(nm + \log(2^{nm}))=O(nm)$. Therefore, we showed that $\BHholes\in \npspace$, then, by Savitch's theorem \cite{savitch}, $\BHholes\in\pspace$.
\end{proof}

Now it remains to show that \BHholes is \pspace-Hard. We do so by exhibiting a polynomial-time reduction from the intersection non-emptiness problem for finite automata.

\subsection{Intersection Non-Emptiness Problem}
Recall that a deterministic finite automaton (DFA for short) is a 5-tuple $(Q, \Sigma, \delta, q_0, F)$, where $Q$ is the set of states, $\Sigma$ is the alphabet, $\delta:Q\times\Sigma\rightarrow Q$ is the transition function, $q_0\in Q$ is the starting state and $F\subseteq Q$ is the set of accepting states. The language accepted by the DFA $A$, called $L(A)$, is the set of strings $x\in\Sigma^*$ such that, starting from $q_0$ and applying $\delta$ repeatedly, $A$ ends up in an accepting state. 

The following is a classic decision problem in automata theory and it was proved to be \pspace-Complete by Kozen \cite{kozen77}:
\begin{definition}[\intersectionnonemptiness]
Given in input a set of $k$ DFAs $\{A_1, A_2, \dots, A_k\}$, with $A_i=(Q_i, \Sigma, \delta_i, q_0^{i}, F_i)$ for $i\in[k]$, say if $\cap_{i=1}^k L(A_i) \neq \varnothing$
\end{definition}

To simplify the exposition, let us assume that all the DFAs have the same set of states $Q$ and the same starting state $q_0$. This is without loss of generality because we can rename the states and add fictitious extra states, without changing the languages of the automata.

\subsection{Reduction Overview}
Consider an instance $I=\{A_1, A_2, \dots, A_k\}$ of the \intersectionnonemptiness, with $A_i=(Q, \Sigma, \delta_i, q_0, F_i)$ for each $i\in[k]$. Without loss of generality, from now on we assume that $|Q|=n$, $Q=\{q_0, q_1, \dots, q_{n-1}\}$, and $|\Sigma|=m$, $\Sigma=\{\sigma_1, \sigma_2, \dots, \sigma_m\}$. Note that our reduction must be polynomial in $k, n,$ and $m$. 

We represent the computation of the DFAs using a string. Specifically, let $\Gamma=\Sigma\cup Q \cup \{\#\}$, where \# is a new symbol, and consider a string of the following form:
\[
R_1\#R_2\#\dots R_r\# \in \Gamma^*
\]
where:
\begin{align*}
& R_j=\sigma^{(j)}q^{(1,j)}q^{(2, j)}\dots q^{(k, j)} & \text{ for }j\in[r],\\
& \sigma^{(j)}\in \Sigma \text{ and } q^{(i,j)}\in Q & \text{ for }i\in[k], j\in[r]
\end{align*}

Call $\rset$ the set of all such strings. Intuitively, $R_j$ contains the $j$-th input symbol and the state of the DFAs after processing the first $j-1$ input symbols.

We say that a string $R_1\#R_2\#\dots R_r\#\in\rset$ is \emph{accepting} if it describes a valid accepting computation for all the automata, that is, if for all $i\in[k]$ the following holds:
\begin{itemize}
    \item $q^{(i,1)}=q_0$
    \item $q^{(i,j+1)}=\delta_i(q^{(i, j)}, \sigma^{(j)})$, for all $j\in[r-1]$
    \item $q^{(i,r)}\in F_i$
\end{itemize}

The main idea of our reduction is to define (i) an encoding of the alphabet $\Gamma$ with strings in $\{\leftc, \rightc, \upc, \downc\}^*$, and (ii) a game level $\mathcal{G}$ of \BHholes, such that, a program solves $\mathcal{G}$ if and only if it is an encoding of some accepting string in $\rset$. This will be enough to conclude our reduction. Indeed, finding an accepting string in $\rset$ is clearly equivalent to finding a string accepted by all the DFAs. More formally:

\begin{observation}\label{obs:nonempty-rset}
Given an instance $\{A_1,A_2, \dots, A_k\}$ of the \intersectionnonemptiness, $\cap_{i=1}^k L(A_i)\neq\varnothing \iff \exists x\in\rset$ accepting
\end{observation}
\begin{proof}
If $\sigma_1\sigma_2\dots\sigma_r\in\cap_{i=1}^k L(A_i)$, then the string $R_1\#R_2\#\dots R_{r+1}\#$ where (i) $\sigma^{(i)}=\sigma_i$ for $i\in[r]$ and $\sigma^{(r+1)}$ is any symbol in $\Sigma$, and (ii) the states are set according to the computation, is an accepting string. Conversely, if $R_1\#R_2\#\dots R_{r}\# \in\rset$ is accepting, then $\sigma^{(1)}\sigma^{(2)}\dots\sigma^{(r-1)}\in \cap_{i=1}^k L(A_i)$ (if $r=1$, the empty string is accepted by all DFAs).
\end{proof}

\subsection{The Encoding}

We first associate an integer value to each element in $Q\cup\Sigma=\{q_0,\dots, q_{n-1}, \sigma_1, \dots, \sigma_m\}$, specifically, we define $\num:Q\cup\Sigma\rightarrow\mathbb{N}$ as:
\begin{align*}
\num(q_i) &= 9\cdot (k+2)\cdot(i+1)& \forall i\in\{0,1,\dots,n-1\} \\
\num(\sigma_i) &= 9\cdot (k+2)\cdot(n+2)\cdot (i+1)& \forall i\in[m]
\end{align*}
Let us now define $\cwset$, the set of \emph{clockwise} strings\footnote{named after the fact that $\rightfull,\downfull,\leftfull,\upfull$ is a clockwise movement}, as:
\[
\cwset = \left\{\rightc^{x_1}\downc^{x_2}\leftc^{x_3}\upc^{x_4} \mid x_1,x_2,x_3,x_4 \geq 6\right\} 
\]
Our encoding will associate to each element of $\Gamma$ a subset of clockwise strings, specifically, let $\enc:\Gamma\rightarrow\mathcal{P}(\cwset)$, where $\mathcal{P}(\cwset)$ is the powerset of set $\cwset$, be defined as:
\begin{align}
\enc(\gamma) &= \left\{\rightc^{\num(\gamma)}\downc^{\num(\gamma)}\leftc^{x_3}\upc^{x_4} \in \cwset \mid x_3,x_4\geq 6\right\}& \forall \gamma\in Q\cup\Sigma \label{eq:encgamma} \\
\enc(\#) &= \left\{\rightc^{w_{\#}}\downc^{x_{\#}}\leftc^{y_{\#}}\upc^{z_{\#}}\right\}& \text{where:} \label{eq:encsharp}\\
w_{\#} &= 9\cdot (k+2)\cdot(n+2)\cdot(m+2) & \nonumber \\ 
x_{\#} &= 18\cdot (k+2)\cdot(n+2)\cdot(m+2) & \nonumber \\ 
y_{\#} &= w_{\#} + 4k + 3 & \nonumber \\ 
z_{\#} &= x_{\#} + 3k + 3 \nonumber & 
\end{align}

Note that the sets are disjoint, therefore it is possible to decode a clockwise string. In particular, if $x\in\enc(\gamma)$, with a slight abuse of notation, we say that $\enc^{-1}(x)=\gamma$. Our goal now is to build the game level $\mathcal{G}$, such that, if $\gamma_1\gamma_2\dots \gamma_t\in \rset$ is accepting, then, the program $x_1x_2\dots x_t$ must solve the level $\mathcal{G}$, where $x_i\in \enc(\gamma_i)$. Conversely, if a program solves $\mathcal{G}$, then it must be a concatenation of clockwise strings $x_1x_2\dots x_t$ such that they can be decoded into $\enc^{-1}(x_i)=\gamma_i$, and $\gamma_1\gamma_2\dots\gamma_t\in\rset$ is accepting. 
\begin{remark*}
At this stage, the numbers used in the encoding might appear arbitrary and obscure. We will point out where these numbers are used as we move forward in the proof.
\end{remark*}

\subsection{The Game Level}

We build several independent sub-levels, each containing a single worker. In particular, we build $2k+4$ sub-levels: $\mathbf{S}=\{CW_1, CW_2, CW_3, \enforcesublevel\} \cup \{M_i^{even}, M_i^{odd}\}_{i\in[k]}$. The final game level $\mathcal{G}$ is created by stacking the sub-levels together and isolating them via holes. Formally, $\mathcal{G}=\stack(\mathbf{S})$.

The three sub-levels $CW_1, CW_2, CW_3$, reported in figure \ref{img:pspace-cw-levels}, are solved by all and only the programs that are concatenations of clockwise strings in $\cwset$ (note that $CW_1$ and $CW_2$ also prevent diagonal movements).

\begin{figure}
\centering
\begin{subfigure}[b]{0.25\textwidth} 
     \centering
     \includegraphics[width=\textwidth]{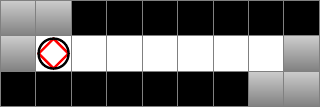}
     \caption{$CW_1$}
 \end{subfigure}
 \begin{subfigure}[b]{0.25\textwidth} 
     \centering
     \includegraphics[height=\textwidth]{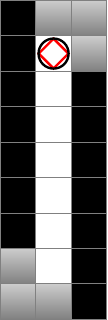}
     \caption{$CW_2$}
 \end{subfigure}
 \begin{subfigure}[b]{0.12\textwidth} 
     \centering
     \includegraphics[width=\textwidth]{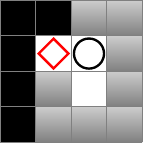}
     \caption{$CW_3$}
 \end{subfigure}
\caption{Sub-levels $CW_1, CW_2, CW_3$, together they ensure that the program is a concatenation of clockwise strings in $\cwset$. Note that $CW_3$ is needed to ensure that the very first movement is \rightfull.}
\label{img:pspace-cw-levels}
\end{figure}

\begin{figure}
\centering
\includegraphics[width=0.12\textwidth]{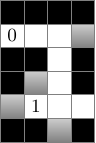} 
\caption{Counter gadget. A worker starting from cell 0 that processes any clockwise string in $\cwset$ ends up in cell 1.}
\label{img:pspace-counter}
\end{figure}

\begin{figure}
\centering
\includegraphics[width=.85\textwidth]{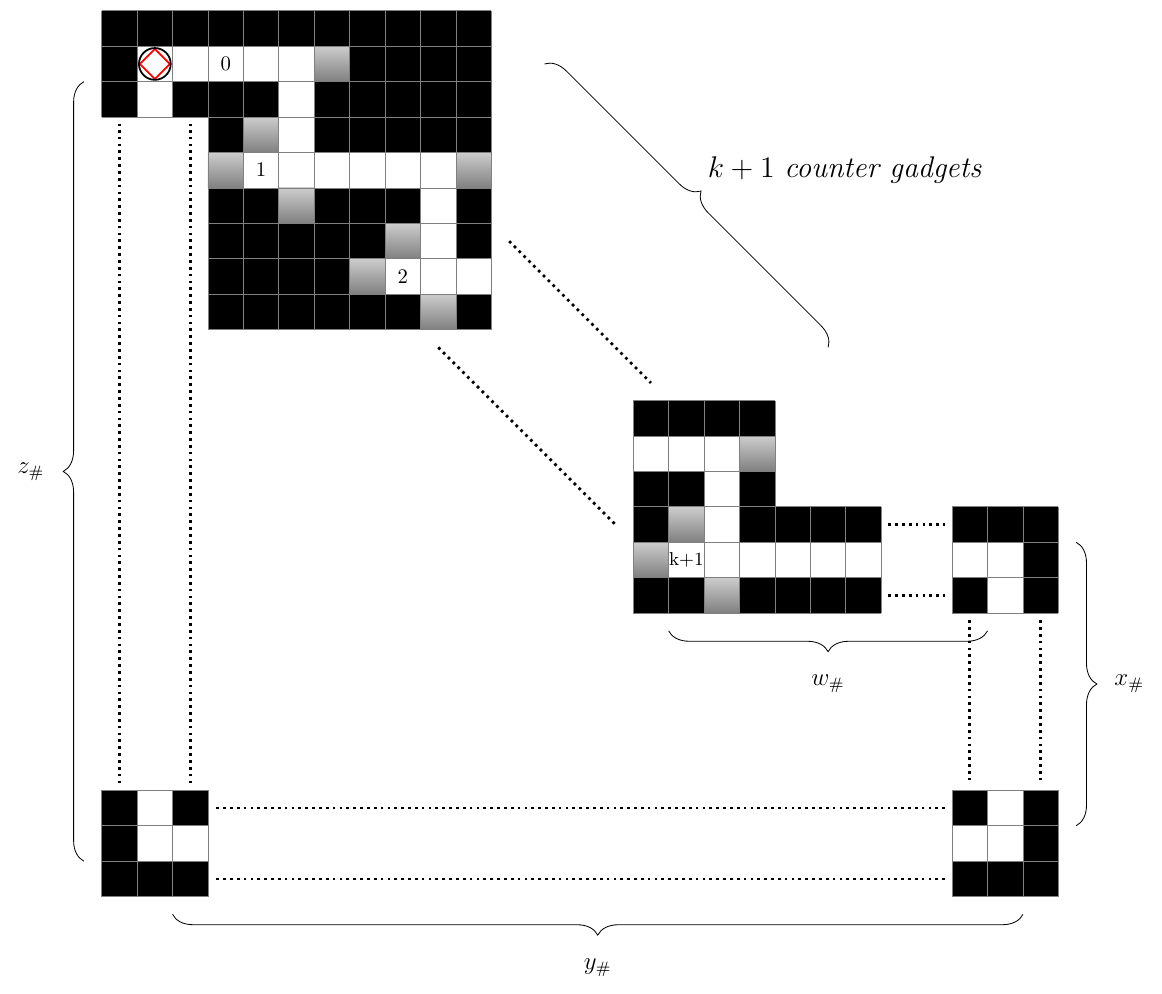} 
\caption{\enforcesublevel sub-level. The values $w_{\#}, x_{\#}, y_{\#}, z_{\#}$ are defined in Equation \eqref{eq:encsharp}. This gadget ensures that accepting programs must be a concatenation of a multiple of $k+2$ clockwise strings each ending with the encoding of $\#$. }
\label{img:pspace-enforcesharp}
\end{figure}

\subsubsection{The Enforce\# Sub-Level}
Let us first introduce the \emph{counter gadget}, in figure \ref{img:pspace-counter}. This gadget is such that, if the worker is standing on the cell labeled with 0, then any clockwise string in $\cwset$ brings the worker to cell 1. Multiple counter gadgets can be concatenated together, and intuitively, these gadgets can be used to \emph{skip} clockwise strings that we do not need to process. 

The \enforcesublevel sub-level, in figure \ref{img:pspace-enforcesharp}, skips the first $k+1$ clockwise strings using $k+1$ counter gadgets, then, it forces the next clockwise string to be $\rightc^{w_{\#}}\downc^{x_{\#}}\leftc^{y_{\#}}\upc^{z_{\#}}$, which is the only encoding of $\#\in\Gamma$, as defined in Equation \eqref{eq:encsharp}. Indeed, if the $(k+2)$th clockwise string was not the encoding of $\#$, the worker would end up in a hole. After processing the whole encoding of $\#$, the worker will be in its starting position again (which is also the only accepting one), ready to possibly process more clockwise strings. 

Therefore, to be more formal, the \enforcesublevel sub-level together with $\{CW_1, CW_2, CW_3\}$, ensures that if $\pi$ is a solving program, then it must be of the form $\overbar{R_1}\overbar{\#}\overbar{R_2}\overbar{\#}\dots\overbar{R_r}\overbar{\#}$, where $\overbar{\#}$ is the only encoding of $\#$, and each $\overbar{R_j}$ is a concatenation of $k+1$ clockwise strings. Note also that all the programs of such form solve these four sub-levels.

\begin{remark*}
The \enforcesublevel sub-level can be built if it holds (i) $y_{\#} = 2 + 4(k+1) + w_{\#} - 3 = w_{\#} + 4k + 3$, and (ii) $z_{\#} = 3(k+1) + x_{\#}$. Note that both these relations are satisfied by our encoding, as reported in Equation \eqref{eq:encsharp}.
\end{remark*}

\subsubsection{The Automata Sub-Levels}
For each DFA $A_i$, $i\in[k]$, we introduce two sub-levels: $M_i^{odd}$ and $M_i^{even}$. Consider a program of the form $\overbar{R_1}\overbar{\#}\dots\overbar{R_r}\overbar{\#}$, where $\overbar{\#}\in\enc(\#)$, and each $\overbar{R_j}$ is a concatenation of $k+1$ clockwise strings. Intuitively, $M_i^{odd}$ will ensure, for each odd $j$, that $\overbar{R_{j+1}}$ follows from $\overbar{R_j}$ according to the computation of $A_i$. $M_i^{even}$ will guarantee the same, but for even $j$'s. These sub-levels will also guarantee that the last state is an accepting one. Therefore, adding all the sub-levels $\{M_i^{odd}, M_i^{even}\}_{i\in[k]}$ will ensure that an accepting program must describe an accepting computation for all the DFAs. 

Before showing the construction of $M_i^{odd}$ and $M_i^{even}$, we introduce three new gadgets.

\begin{figure}[h!]
\centering
\includegraphics[width=0.85\textwidth]{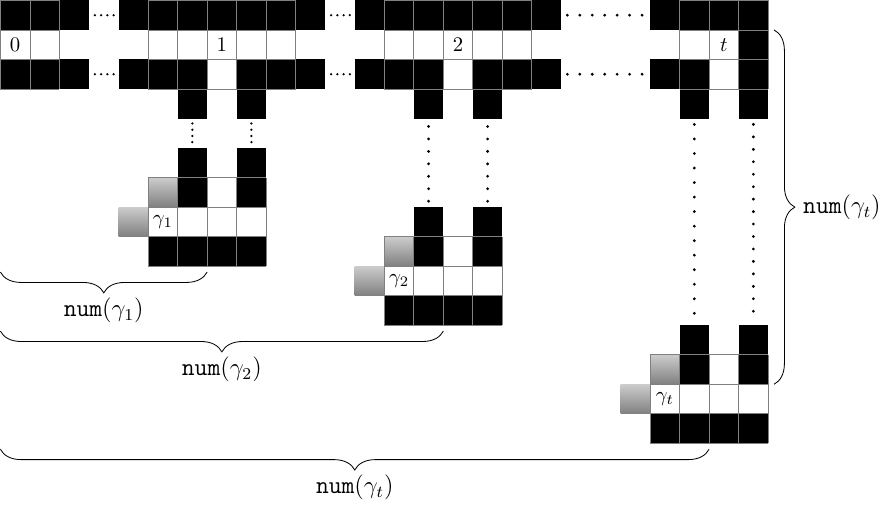} 
\caption{The $S$-\emph{selector gadget} for $S=\{\gamma_1, \gamma_2, \dots, \gamma_t\}\subseteq Q\cup\Sigma$. A clockwise string reaches the cell labeled with $\gamma_{\ell}$ if it moves \rightfull and \downfull exactly $\num(\gamma_{\ell})$ times.}\label{img:pspace-selector}
\end{figure}

The $S$-\emph{selector gadget}, in figure \ref{img:pspace-selector}, is parametrized by a set $S=\{\gamma_1, \gamma_2, \dots, \gamma_t\} \subseteq Q\cup\Sigma$ such that $\num(\gamma_j) < \num(\gamma_{j+1})$. If the worker is standing on the cell labeled with 0 of the selector gadget, then the next clockwise string must be any encoding of any element $\gamma_{\ell} \in S$, which will lead the worker to the cell labeled with $\gamma_{\ell}$. Indeed, from our encoding in Equation \eqref{eq:encgamma}, to verify that a clockwise string $\rightc^{x_1}\downc^{x_2}\leftc^{x_3}\upc^{x_4}\in \cwset$ is an encoding of a certain element $\gamma_{\ell}\in Q\cup\Sigma$, it suffices to check that $x_1=x_2=\num(\gamma_{\ell})$, which is what the gadget does. It is also easy to check that the worker will fall into a hole if the clockwise string is not an encoding of any element of $S$. This gadget will be useful for choosing the next input symbol and performing different checks depending on the current state of the automaton.

\begin{figure}[h!]
\centering
\includegraphics[width=.3\textwidth]{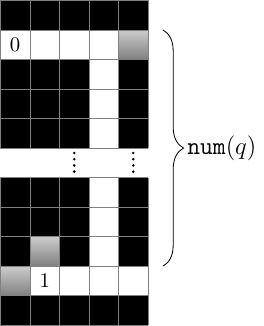}
\caption{The $q$-\emph{forcer gadget} for $q \in Q$. Starting from cell 0, the worker reaches cell 1 with a clockwise string if and only if the number of $\downc$ symbols is equal to $\num(q)$.}\label{img:pspace-forcer}
\end{figure}

For a state $q\in Q$, the $q$-\emph{forcer gadget}, in figure \ref{img:pspace-forcer}, forces the next clockwise string to have a number of \downfull steps equal $\num(q)$. Therefore, if we know that $c\in\cwset$ is an encoding of some state, then, by using the $q$-forcer gadget we impose the constraint that $c$ must be an encoding of $q\in Q$. This gadget will be useful to force the string to respect the transition function.

\begin{figure}[h!]
\centering
\includegraphics[width=0.8\textwidth]{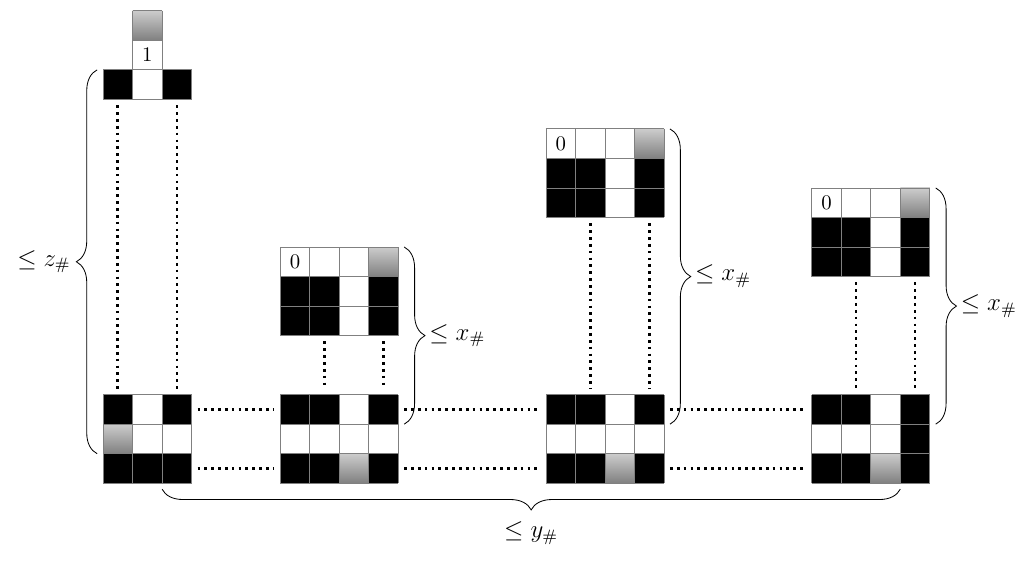}
\caption{The \emph{go-back gadget}. Starting from any cell 0, the clockwise string $\rightc^{w_{\#}}\downc^{x_{\#}}\leftc^{y_{\#}}\upc^{z_{\#}}$, as described in Equation \eqref{eq:encsharp}, brings the worker back to the cell 1. Note that more cells 0 can be added as needed, provided that the highlighted constraints are respected.}\label{img:pspace-goback}
\end{figure}

The last gadget we need is the \emph{go-back gadget}, in figure \ref{img:pspace-goback}. Suppose the worker is in one of the cells labeled with 0, then, by processing the only encoding of $\#\in\Gamma$, it will go back to cell 1. This gadget will be useful to "go back" to a selector gadget and analyze the next chunk of the program.

Note that these three gadgets and the counter gadget can be concatenated together. The only precaution to take is that when concatenating a selector gadget to a counter gadget, the cell 0 of the selector (figure \ref{img:pspace-selector}) must coincide with the cell 1 of the counter (figure \ref{img:pspace-counter}). Similarly, when concatenating a \emph{go-back gadget} to a selector gadget, the cell $1$ of the \emph{go-back} must coincide with the cell $0$ of the selector.

\begin{figure}[h!]
\centering
\includegraphics[width=0.9\textwidth]{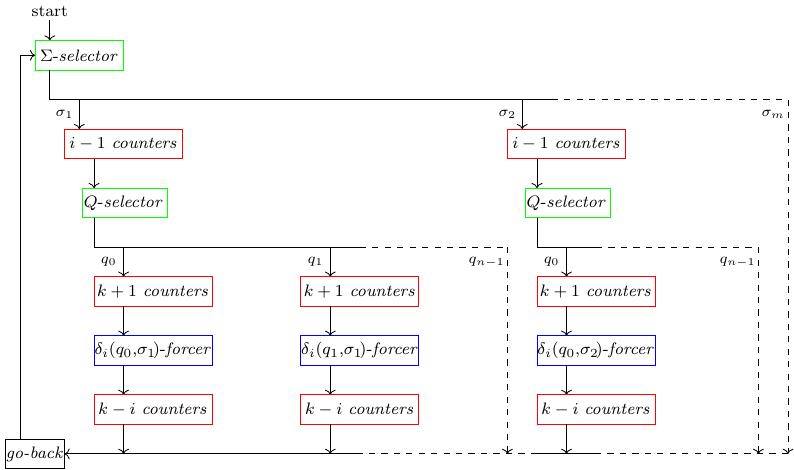}
\caption{$M_i^{odd}$ sub-level associated to the DFA $A_i$, $i\in[k]$. The starting position of the single worker is at the beginning of the $\Sigma$-selector (i.e., the cell 0 of the $\Sigma$-selector, taking figure \ref{img:pspace-selector} as a reference). The first cell of the $\Sigma$-selector is also an accepting cell. Moreover, each branch of the $Q$-selectors corresponding to an accepting state $q_{acc}\in F_i$, contains an accepting cell after the first $k+1-i$ counters (i.e., the accepting cell is the cell labeled with 1, taking figure \ref{img:pspace-counter} as a reference, in the $(k+1-i)$th of the $k+1$ counters).}\label{img:pspace-Mi-odd}
\end{figure}

\begin{figure}[h!]
\centering
\includegraphics[height=0.30\textwidth]{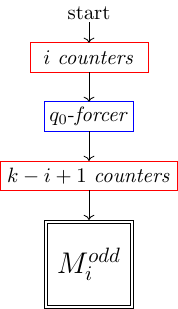}
\caption{$M_i^{even}$ sub-level associated to the DFA $A_i$, $i\in[k]$. The starting position of the single worker is at the beginning of the first counter. The accepting cells are the same described for the $M_i^{odd}$ sub-level.}\label{img:pspace-Mi-even}
\end{figure}

Fixed $i\in[k]$, we report, in figure \ref{img:pspace-Mi-odd}, the sub-level $M_i^{odd}$, and, in figure \ref{img:pspace-Mi-even} the sub-level $M_i^{even}$. Both are represented in a schematic way. Both sub-levels use the go-back gadget, which requires processing exactly the encoding of $\#$: this is guaranteed by the \enforcesublevel sub-level.

We now argue that the sub-levels $\{CW_1, CW_2, CW_3, \enforcesublevel, M_i^{odd}, M_i^{even}\}$ are solved by all and only the programs describing a valid accepting computation of DFA $A_i$.

We know, from $\{CW_1, CW_2, CW_3, \enforcesublevel\}$, that a solving program $\pi$ must of the form: $\overbar{R_1}\overbar{\#}\overbar{R_2}\overbar{\#}\dots\overbar{R_r}\overbar{\#}$, where $\overbar{\#}\in\enc(\#)$, and $\overbar{R_j}=c_1^j c_2^j\dots c_{k+1}^j$ is a concatenation of $k+1$ clockwise strings. Let us first show that $\pi$ can always be decoded:
\begin{itemize}
\item The $\Sigma$-selector of $M_i^{odd}$ (resp. $M_i^{even}$) ensures that the first symbol of each $\overbar{R_j}$, for odd $j$ (resp. even $j$), is the encoding of a symbol in $\Sigma$. Therefore, for all $j\in[r]$, it must be $c_1^j\in\enc(\sigma)$, for some $\sigma\in\Sigma$.
\item Similarly, the $Q$-selectors of $M_i^{odd}$ (resp. $M_i^{even}$) ensure that the $(i+1)$th symbol of each $\overbar{R_j}$, for odd $j$ (resp. even $j$), is the encoding of a state. Therefore, for all $j\in[r]$, it must be $c_{i+1}^j\in\enc(q)$, for some $q\in Q$.
\end{itemize}
Moreover, $\pi$ describes a valid computation:
\begin{itemize}
\item the first forcer of $M_i^{even}$ ensures that the computation starts from the starting state $q_0$: $c_{i+1}^1\in\enc(q_0)$
\item Suppose that the computation is valid up to $\overbar{R_j}$, that is, the state $q=\enc^{-1}(c_{i+1}^j)$ is correctly reached from $q_0$. Let us assume that $j$ is odd. Then, $M_i^{odd}$ will follow the branch $\sigma$ in the $\Sigma$-selector, for some $\sigma\in\Sigma$, and the branch $q$ in the $Q$-selector. Then, the $\delta_i(q, \sigma)$-forcer in $M_i^{odd}$ forces $c_{i+1}^{j+1}$ to be in $\enc(\delta_i(q, \sigma))$, therefore, the state reached in $\overbar{R_{j+1}}$ is correct too. If instead $j$ is even, $M_i^{even}$ ensures that the state reached in $\overbar{R_{j+1}}$ is correct.  
\end{itemize}

It is left to show that the computation described by $\pi$ is accepting for $A_i$. Suppose $r$ is odd. Then, at the end of the computation, the worker of $M_i^{even}$ will be at the beginning of the $\Sigma$-selector (which is an accepting cell), and the worker of $M_i^{odd}$, which is "shifted forward" by $k+2$ clockwise strings, will be at the end of the $(k+1-i)$th counter right after the $Q$-selector, but note that there is an accepting cell in such position only if the worker is in a branch corresponding to a state $q_{acc}\in F_i$, therefore $c_{i+1}^r\in \enc(q_{acc})$ and the computation is accepting. If $r$ is even, the situation is analogous: the worker of $M_i^{odd}$ is at the beginning of the $\Sigma$-selector, but the one of $M_i^{even}$ is at the end of the $(k+1-i)$th counter after a $Q$-selector, and therefore it must be in an accepting branch.

Moreover, one can easily see that any encoding of an accepting computation solves all the sub-levels. Indeed, the $\Sigma$-selector allows the user to choose the input string, and then, since the computation is accepting, one of $M_i^{even}$ and $M_i^{odd}$ will end up in the $(k+1-i)$th counter of an accepting branch and the other will stay at the beginning of the $\Sigma$-selector. Therefore, all the sub-levels would be solved.

\begin{remark*}
Our encoding allows the construction of $M_i^{odd}$ and $M_i^{even}$. First, from our encoding, we have that for $j\in\{0,1,\dots,n-2\}$, $\num(q_{j+1}) - \num(q_j)=9(k+2)$, while the space actually needed between two branches of the $Q$-selector is: $4(k+1) + 5 + 4(k-i) + 4 + 6 < 8(k+1) + 11 \leq 9(k+2)$ where the last 6 takes into account the overhead of the selector. 
Similarly, for $j\in[m-1]$, $\num(\sigma_{j+1}) - \num(\sigma_j) = 9(k+2)(n+2)$, while the space required between two branches of the $\Sigma$-selector is at most $4(i-1) + 9(k+2)(n+1) + 6 < 9(k+2)(n+2)$. To conclude, we need only to check that the constraints for the go-back gadget are satisfied. The width of the whole $M_i^{odd}$ sub-level is at most $9(k+2)(n+2)(m+2) = w_{\#} < y_{\#}$. The height of the sub-level is instead at most $\num(\sigma_m) + 4(i-1) + \num(q_{n-1}) + 4(2k+1-i) + \num(q_{n-1}) + 3 \leq 18(k+2)(n+2)(m+1) + 8k < 18(k+2)(n+2)(m+2) = x_{\#} < z_{\#}$. Therefore all the gadgets can be concatenated together.
\end{remark*}

\subsection{Conclusion of the proof}
Putting all the pieces together, we can prove the theorem.
\begin{proof}[Proof of Theorem \ref{thm:pspace-comp}]
The problem is in \pspace by Observation \ref{obs:membership-pspace}. Consider the following reduction from the \intersectionnonemptiness: an instance $I=\{A_1, A_2, \dots, A_k\}$ is associated with the level $\mathcal{G}=\stack(\{CW_1, CW2, CW_3, \enforcesublevel\} \cup \{M_i^{odd}, M_i^{even}\}_{i\in[k]})$. 

Suppose $I$ can be solved, then, by Observation \ref{obs:nonempty-rset}, there exists an accepting string $x_1x_2\dots x_t\in\rset$. Consider any encoding $\pi=y_1y_2\dots y_t$ such that $y_j\in\enc(x_j)$. Such program can be rewritten in the form $\overbar{R_1}\overbar{\#}\dots\overbar{R_r}\overbar{\#}$. Therefore, as we argued, $\pi$ solves $\{CW_1, CW_2, CW_3, \enforcesublevel\}$, and, given that $I$ is solved, for each $i\in[k]$, $\pi$ describes an accepting computation for $A_i$, then, it also solves $M_i^{odd}$ and $M_i^{even}$. Therefore, $\mathcal{G}$ is solved: all its workers will be standing on an accepting cell at the end of the program.

Suppose now that there exists a program $\pi\in\{\leftc, \rightc, \upc, \downc\}^*$ solving $\mathcal{G}$. As we argued, such a program can be decoded into a string $R_1\#\dots R_r\#\in\rset$. Given that, for each $i\in[k]$, $M_i^{odd}$ and $M_i^{even}$ are solved, it means that the string represents an accepting computation for $A_i$. That is, letting $R_j=\sigma^{(j)}q^{(1,j)}\dots q^{(k, j)}$, it holds: (i) $q^{(i, 1)}=q_0$, (ii) $q^{(i, j+1)}=\delta_i(q^{(i, j)}, \sigma^{(j)})$ for $j\in[r-1]$, and (iii) $q^{(i, r)}\in F_i$. Therefore, the string in $\rset$ is accepting. Using Observation \ref{obs:nonempty-rset}, it follows that $I$ is solvable.

Finally, observe that the number of cells in each sub-level is at most $O(y_{\#}\cdot z_{\#}) = O(k^2 n^2 m^2)$, and since there are $2k+4$ sub-levels, $\mathcal{G}$ has at most $O(k^3 n^2 m^2)$ cells, therefore the reduction can be carried out in polynomial time.
\end{proof}

\section{Conclusions}\label{sec:conclusion}
We analyzed the computational complexity of the video game ``7 Billion Humans''. The game involves controlling multiple workers simultaneously to direct them to some destination cells. When each cell is either empty or contains a wall, the problem of deciding if a level is solvable is \np-Hard, while adding holes makes the problem \pspace-Complete. We always restricted to instances where each connected component contains at most one worker. We also observed that the simple structure of our reductions entails hardness results for other problems. 

While \BHessential is \np-Hard and clearly in \pspace, it is not known whether it lies in \np. We stress that in \BHessential each connected component can contain at most one worker. We leave the problem of completely characterizing \BHessential as an interesting open problem.

\bibliography{7bh}

\end{document}